\documentclass[proceedings]{stacs}
\stacsheading{2009}{301--312}{Freiburg}
\firstpageno{301}

\usepackage{amsfonts,color,graphicx,ifpdf,latexsym,url,wrapfig}
\usepackage
  [a4paper,breaklinks,bookmarks,bookmarksnumbered,bookmarksopen,bookmarksopenlevel=2]
  {hyperref}
{\makeatletter \hypersetup{pdftitle={\@title}}}
\hypersetup{pdfauthor={}}

\makeatletter
\def\GrabProofArgument[#1]{ #1: \egroup\ignorespaces}
\def\proof{\noindent\emph\bgroup Proof%
           \@ifnextchar[{\GrabProofArgument}{. \egroup\ignorespaces}}

\makeatother

\let\epsilon=\varepsilon

\def\Dist{\mathop{\mathrm{Dist}}\nolimits}

\def\Cooperative{Cooperative}
\def\edgecost{creation cost}
\def\pathcost{usage cost}

\newcommand{\polylg}{\mathop{\rm polylg}\nolimits}

\begin{document}

\title[The Price of Anarchy in Cooperative Network Creation Games]
      {The Price of Anarchy in \\ Cooperative Network Creation Games}

\author[MIT]{E. D. Demaine}{Erik D. Demaine}
\address[MIT]{MIT Computer Science and Artificial Intelligence Laboratory,
              32 Vassar St., Cambridge, MA 02139, USA}
\email{edemaine@mit.edu}

\author[MIT,ATT]{M. Hajiaghayi}{MohammadTaghi Hajiaghayi}
\address[ATT]{AT\&T Labs --- Research, 180 Park Ave., Florham Park, NJ 07932,
              USA}
\email{hajiagha@research.att.com}

\author[ITPM,Sharif]{H. Mahini}{Hamid Mahini}
\address[ITPM]{School of Computer Science, Institute for Theoretical Physics
               and Mathematics, Tehran}
\email{mahini@ce.sharif.edu}

\author[Sharif]{M. Zadimoghaddam}{Morteza Zadimoghaddam}
\address[Sharif]{Department of Computer Engineering,
                 Sharif University of Technology}
\email{zadimoghaddam@ce.sharif.edu}

\begin{abstract}
We analyze
the structure of equilibria and the price of anarchy in the family of
network creation games considered extensively in the past few years,
which attempt to unify the network design and network routing problems
by modeling both creation and usage costs.
In general, the games are played on a host graph,
where each node is a selfish independent agent (player)
and each edge has a fixed link creation cost~$\alpha$.
Together the agents create a network (a subgraph of the host graph)
while selfishly minimizing the link creation costs plus
the sum of the distances to all other players (usage cost).
In this paper, we pursue two important facets of the network creation~game.

First, we study extensively a natural version of the game, called
the cooperative model, where nodes can collaborate and share the cost of
creating any edge in the host graph.  We prove the first nontrivial bounds
in this model, establishing that the price of anarchy is
polylogarithmic in $n$ for all values of~$\alpha$ in complete host graphs.
This bound is the first result of this type for any version of the network
creation game; most previous general upper bounds are polynomial in~$n$.
Interestingly, we also show that equilibrium graphs have polylogarithmic
diameter for the most natural range of~$\alpha$ (at most $n \polylg n$).

Second, we study the impact of the natural assumption that the host graph
is a general graph, not necessarily complete.  This model is a simple example
of nonuniform creation costs among the edges (effectively allowing weights of
$\alpha$ and~$\infty$).
We prove the first assemblage of upper and lower bounds for this context,
establishing nontrivial tight bounds for many ranges of~$\alpha$,
for both the unilateral and cooperative versions of network creation.
In particular, we establish polynomial lower bounds for both versions and many
ranges of~$\alpha$, even for this simple nonuniform cost model, which sharply
contrasts the conjectured constant bounds for these games in complete
(uniform) graphs.
\end{abstract}

\maketitle

\section{Introduction}

A fundamental family of problems at the intersection between computer science
and operations research is \emph{network design}.  This area of research has
become increasingly important given the continued growth of computer networks
such as the Internet.
Traditionally, we want to find a minimum-cost (sub)network that
satisfies some specified property such as $k$-connectivity
or connectivity on terminals (as in the classic Steiner tree problem).
This goal captures the (possibly incremental) creation cost of
the network, but does not incorporate the cost of actually using the network.
In contrast, \emph{network routing} has the goal of optimizing the usage cost
of the network, but assumes that the network has already been created.

\emph{Network creation games} attempt to unify the network design
and network routing problems by modeling both creation and usage costs.
In general, the game is played on a \emph{host graph},
where each node is an independent agent (player),
and the goal is to create a network from a subgraph of the host graph.
Collectively, the nodes decide which edges of the host graph are worth
creating as links in the network.
Every link has the same creation cost~$\alpha$.
(Equivalently, links have creation costs of $\alpha$ and~$\infty$,
depending on whether they are edges of the host graph.)
In addition to these creation costs, each node incurs a usage cost
equal to the sum of distances to all other nodes in the network.
Equivalently, if we divide the cost (and thus~$\alpha$) by the number $n$ of
nodes, the usage cost for each node is its average distance to all other nodes.
(This natural cost model has been used in, e.g., contribution games and
network-formation games.)

There are several versions of the network creation game that vary how
links are purchased.  In the \emph{unilateral} model---introduced by
Fabrikant, Luthra, Maneva, Papadimitriou, and Shenker \cite{FLMPS03}---every
node (player) can locally decide to purchase any edge incident
to the node in the host graph, at a cost of~$\alpha$.
In the \emph{bilateral} model---introduced
by Corbo and Parkes \cite{CP05}---both endpoints of an edge must agree
before they can create a link between them, and the two nodes share the
$\alpha$ creation cost equally.
In the \emph{cooperative} model---%
introduced by Albers, Eilts, Even-Dar, Mansour, and Roditty~\cite{AEEMR06}---any
node can purchase any amount of any edge in the host graph,
and a link gets created when the total purchased amount is at least~$\alpha$.

To model the dominant behavior of large-scale networking scenarios
such as the Internet, we consider the case where every node (player)
selfishly tries to minimize its own creation and usage cost
\cite{Jackson03,FLMPS03,AEEMR06,CP05}.
This game-theoretic setting naturally leads to the various kinds of
\emph{equilibria} and the study of their structure.
Two frequently considered notions are \emph{Nash equilibrium}
\cite{Nash50,Nash51}, where no player can change its strategy
(which edges to buy) to locally improve its cost, and
\emph{strong Nash equilibrium} \cite{Aumann59,AFM07,Albers08},
where no coalition of players
can change their collective strategy to locally improve the cost
of each player in the coalition.
Nash equilibria capture the combined effect of both selfishness and lack of
coordination, while strong Nash equilibria separates these issues,
enabling coordination and capturing the specific effect of selfishness.
However, the notion of strong Nash equilibrium is extremely restrictive
in our context,
because all players can simultaneously change their entire strategies,
abusing the local optimality intended by original Nash equilibria,
and effectively forcing globally near-optimal solutions
\cite{AFM07}.

We consider weaker notions of equilibria, which broadens the scope of
equilibria and therefore strengthens our upper bounds,
where players can change their strategy on only a single edge at a time.
In a \emph{collaborative equilibrium}, even coalitions of players do
not wish to change their collective strategy on any single edge;
this concept is particularly important for the cooperative network
creation game, where multiple players must negotiate their relative
valuations of an edge.
(This notion is the natural generalization of pairwise stability
from \cite{CP05} to arbitrary cost sharing.)
Collaborative equilibria are essentially a compromise between Nash
and strong Nash equilibria: they still enable coordination among players
and thus capture the specific effect of selfishness, like strong Nash,
yet they consider more local moves, in the spirit of Nash.
In particular, any results about all collaborative equilibria also apply
to all strong Nash equilibria.
Collaborative equilibria also make more sense computationally:
players can efficiently detect equilibrium using a simple bidding procedure
(whereas this problem is NP-hard for strong Nash),
and the resulting dynamics converge to such equilibria
(see Section~\ref{dynamics}).

The structure of equilibria in network creation games is not very well
understood.  For example, Fabrikant et al.~\cite{FLMPS03} conjectured
that equilibrium graphs in the unilateral model were all trees,
but this conjecture was disproved by Albers et al.~\cite{AEEMR06}.
One particularly interesting structural feature is whether all equilibrium
graphs have small \emph{diameter} (say, polylogarithmic),
analogous to the small-world phenomenon \cite{Kleinberg00,EK06},
In the original unilateral version of the problem, the best general lower
bound is just a constant and the best general upper bound is polynomial.
A closely related issue is the \emph{price of anarchy}
\cite{KP99,papa01,roughgarden-phd},
that is, the worst possible ratio of the total cost of an equilibrium
(found by independent selfish behavior) and the optimal total cost possible
by a centralized solution (maximizing social welfare).
The price of anarchy is a well-studied concept in algorithmic game theory
for problems such as load balancing, routing, and network design;
see, e.g.,
\cite{papa01,CV02,Roughgarden02,FLMPS03,ADTW03,ADKTW04,CFSK04,CP05,AEEMR06,PODC07}.
Upper bounds on diameter of equilibrium graphs translate to approximately
equal upper bounds on the price of anarchy, but not necessarily vice versa.
In the unilateral version, for example, there is a general $2^{O(\sqrt{\lg n})}$
upper bound on the price of anarchy.

\paragraph{\bf Previous work.}
Network creation games have been studied extensively in the literature
since their introduction in 2003.

For the unilateral version and a complete host graph,
Fabrikant et al.~\cite{FLMPS03} prove an upper bound of $O(\sqrt{\alpha})$
on the price of anarchy for all~$\alpha$.
Lin \cite{Lin03} proves that the price of anarchy is constant
for two ranges of $\alpha$: $\alpha = O(\sqrt n)$ and
$\alpha \geq c \, n^{3/2}$ for some $c > 0$.
Independently, Albers et al.~\cite{AEEMR06} prove that
the price of anarchy is constant for $\alpha = O(\sqrt n)$,
as well as for the larger range
$\alpha \geq 12 \, n \lceil \lg n \rceil$.
In addition, Albers et al.\ prove a general upper bound of $15
\left(1+(\min\{\frac{\alpha^2}{n},\frac{n^2}{\alpha}\})^{1/3}\right)$.
The latter bound shows the first sublinear worst-case bound, $O(n^{1/3})$,
for all~$\alpha$.
Demaine et al.~\cite{PODC07} prove the first $o(n^\epsilon)$
upper bound for general~$\alpha$, namely, $2^{O(\sqrt{\lg n})}$.
They also prove a constant upper bound for $\alpha = O(n^{1-\epsilon})$
for any fixed $\epsilon > 0$,
and improve the constant upper bound by Albers et al.\
(with the lead constant of~$15$) to $6$ for $\alpha < (n/2)^{1/2}$ and
to $4$ for $\alpha < (n/2)^{1/3}$.
Andelmen et al.~\cite{AFM07} show that, among strong Nash equilibria,
the price of anarchy is at most~$2$.

For the bilateral version and a complete host graph,
Corbo and Parkes \cite{CP05} prove that the price of anarchy is
between $\Omega(\lg \alpha)$ and $O(\min\{\sqrt \alpha,n/\sqrt \alpha)$.
Demaine et al.~\cite{PODC07} prove that the upper bound is tight,
establishing the price of anarchy to be $\Theta(\min\{\sqrt{\alpha},
n/\sqrt{\alpha}\})$ in this case.

For the cooperative version and a complete host graph,
the only known result is an upper bound of
$15 \left(1+(\min\{\frac{\alpha^2}{n},\frac{n^2}{\alpha}\})^{1/3}\right)$,
proved by Albers et al.~\cite{AEEMR06}.

Other variations of network creation games allow nonuniform interests
in connectivity between nodes \cite{HM07} and nodes with limited budgets
for buying edges \cite{LPRST08}.

\paragraph{\bf Our results.}

Our research pursues two important facets of the network creation game.

First, we make an extensive study of a natural version of the game---the
cooperative model---where the only previous results were simple extensions
from unilateral analysis.  We substantially improve the bounds in this case,
showing that the price of anarchy is polylogarithmic in $n$ for \emph{all
values of~$\alpha$} in complete graphs.  This is the first result of this type
for any version of the network creation game.
As mentioned above, this result applies to both collaborative equilibria
and strong Nash equilibria.
Interestingly, we also show that equilibrium graphs have polylogarithmic
diameter for the most natural range of~$\alpha$ (at most $n \polylg n$).
Note that, because of the locally greedy nature of Nash equilibria,
we cannot use the classic probabilistic spanning (sub)tree embedding machinery
of \cite{Bartal98,FRT04,EEST05}
to obtain polylogarithmic bounds (although this machinery can be applied
to approximate the global social optimum).

Second, we study the impact of the natural assumption that the host graph
is a general graph, not necessarily complete, inspired by practical
limitations in constructing network links.  This model is a simple example
of nonuniform creation costs among the edges (effectively allowing weights of
$\alpha$ and~$\infty$).  Surprisingly, no bounds on the diameter or
the price of anarchy have been proved before in this context.
We prove several upper and lower bounds, establishing nontrivial
tight bounds for many ranges of~$\alpha$, for both the unilateral and
cooperative versions.
In particular, we establish polynomial lower bounds for both versions and
many ranges of~$\alpha$, even for this simple nonuniform cost model.
These results are particularly interesting because,
by contrast, no superconstant lower bound has been shown for either game
in complete (uniform) graphs.
Thus, while we believe that the price of anarchy is polylogarithmic
(or even constant) for complete graphs, we show a significant departure
from this behavior in general graphs.

Our proof techniques are most closely related in spirit to
``region growing'' from approximation algorithms; see, e.g.,
\cite{LR99}.  Our general goal is to prove an upper bound on diameter
by way of an upper bound on the expansion of the graph.  However,
we have not been able to get such an argument to work directly in general.
The main difficulty is that, if we imagine building a breadth-first-search tree
from a node, then connecting that root node to another node does not
necessarily benefit the node much: it may only get closer to a small fraction
of nodes in the BFS subtree.  Thus, no node is motivated selfishly to improve
the network, so several nodes must coordinate their changes to make
improvements.  The cooperative version of the game gives us some leverage
to address this difficulty.  We hope that this approach, particularly the
structure we prove of equilibria, will shed some light on the still-open
unilateral version of the game, where the best bounds on the price of anarchy
are $\Omega(1)$ and $2^{O(\sqrt{\lg n})}$.

Table~\ref{summary} summarizes our results.
Section~\ref{Cooperative Version in Complete Graphs} proves our
polylogarithmic upper bounds on the price of anarchy
for all ranges of $\alpha$ in the cooperative network creation
game in complete graphs.
Section~\ref{Cooperative Version in General Graphs} considers how
the cooperative network creation game differs in general graphs,
and proves our upper bounds for this model.
Section~\ref{Unilateral Version in General Graphs} extends these
results to apply to the unilateral network creation game
in general graphs.
Section~\ref{Lower Bounds in General Graphs} proves lower bounds for
both the unilateral and cooperative network creation games
in general graphs, which match our upper bounds for some ranges of~$\alpha$.

\begin{table*}
  \centering
  \footnotesize
  \tabcolsep=0pt
  \def\LABEL#1{\hbox to 0pt{\hss#1\hss}}
  \def\BOUND#1#2{\multicolumn{#1}{|c|}{\hbox{\,}#2\hbox{\,}}}
  \hbox to \hsize{\hss
  \begin{tabular}{lclclclclclclclclclc}
    \multicolumn{1}{r}{$\alpha = $ \hbox{~~}}
    &\LABEL{$0$}& \hbox{\qquad}
    &\LABEL{$n$}& \hbox{\qquad\qquad}
    &\LABEL{$n \lg^{0.52} n$}& \hbox{\qquad\qquad}
    &\LABEL{$n \lg^{7.16} n$\hspace{1em}}& \hbox{\qquad\quad}
    &\LABEL{$n^{3/2}$}& \hbox{\qquad}
    &\LABEL{\hspace*{1em}$n^{5/3}$}& \hbox{\qquad\qquad}
    &\LABEL{$n^2$} & \hbox{\qquad\qquad}
    &\LABEL{$n^2 \lg n$} & \hbox{\qquad}
    &\LABEL{$\infty$}
    \\ \cline{2-17}
   \Cooperative, complete graph\hbox{~} && \BOUND{1}{$\Theta(1)$} && \BOUND{1}{$\lg^{3.32} n$} && \BOUND{1}{$O\big(\lg n {+} \sqrt{n \over \alpha} \lg^{3.58}n\big)$} &&  \BOUND{9}{$\Theta(1)$}
    \\ \cline{2-17}
    \Cooperative, general graph\hbox{~} && \BOUND{1}{$O(\alpha^{1/3})$} && \BOUND{7}{$O(n^{1/3})$, $\Omega(\sqrt{{\alpha \over n}})$}  &&
\BOUND{1}{$\Theta({n^2 \over \alpha})$}&&
\BOUND{1}{$O\big({n^2 \over \alpha}\lg n\big)$} && \BOUND{1}{$\Theta(1)$}
    \\ \cline{2-17}
    Unilateral, general graph\hbox{~} && \BOUND{1}{$O(\alpha^{1/2})$} && \BOUND{5}{$O(n^{1/2})$, $\Omega({\alpha \over n})$} && \BOUND{3}{$\Theta({n^2 \over \alpha})$} && \BOUND{3}{$\Theta(1)$}
    \\ \cline{2-17}
  \end{tabular}
  \hss
  }
  \caption{Summary of our bounds on equilibrium diameter and price of anarchy
    for cooperative network creation in complete graphs,
    and unilateral and cooperative network creation in general graphs.
    For all three of these models, our bounds are strict
    improvements over the best previous bounds.}
  \label{summary}
\end{table*}

\section{Models}

In this section, we formally define the different models
of the network creation game.

\subsection{Unilateral Model}

We start with the unilateral model, introduced in~\cite{FLMPS03}.
The game is played on a \emph{host graph} $G = (V,E)$.
Assume $V = \{1, 2, \dots, n\}$.
We have $n$ players, one per vertex.
The strategy of player $i$ is specified by a subset $s_i$ of
$\{j : \{i,j\} \in E\}$, defining the set of neighbors
to which player $i$ creates a link.
Thus each player can only create links corresponding to edges incident to
node $i$ in the host graph $G$
Together, let $s = \langle s_1, s_2, \dots, s_n \rangle$
denote the joint strategy of all players.

To define the cost of strategies,
we introduce a spanning subgraph $G_s$ of the host graph~$G$.
Namely, $G_s$ has an edge $\{i,j\} \in E(G)$ if either $i \in s_j$
or $j \in s_i$.
Define $d_{G_s}(i,j)$ to be the distance between vertices $i$ and $j$
in graph~$G_s$.
Then the cost incurred by player $i$ is
$
c_i(s) = \alpha \, |s_i| + \sum_{j=1}^n d_{G_s}(i,j).
$
The total cost incurred by joint strategy $s$ is
$c(s) = \sum_{i=1}^n c_i(s)$.

A (pure) \emph{Nash equilibrium} is a joint strategy $s$
such that $c_i(s) \leq c_i(s')$ for all joint strategies $s'$
that differ from $s$ in only one player~$i$.
The \emph{price of anarchy} is then the maximum cost of a Nash equilibrium
divided by the minimum cost of any joint strategy
(called the \emph{social optimum}).

\subsection{\Cooperative\ Model}

Next we turn to the cooperative model, introduced in \cite{FLMPS03,AEEMR06}.
Again, the game is played on a host graph $G = (V,E)$,
with one player per vertex.
Assume $V = \{1, 2, \dots, n\}$ and $E = \{e_1, e_2, \dots, e_{|E|}\}$.
Now the strategy of player $i$ is specified by a vector
$s_i = \langle s(i, e_1), s(i,e_2), \dots, s(i, e_{|E|}) \rangle$,
where $s(i, e_j)$ corresponds to the value that player $i$
is willing to pay for link~$e_j$.
Together, $s = \langle s_1, s_2, \dots, s_n \rangle$
denotes the strategies of all players.

We define a spanning subgraph $G_s = (V,E_s)$ of the host graph~$G$:
$e_j$ is an edge of $G_s$ if $\sum_{i \in V(G)} s(i,e_j) \geq \alpha$.
To make the total cost for an edge $e_j$ exactly $0$ or $\alpha$ in all cases,
if $\sum_{i \in V(G)} s(i,e_j) > \alpha$, we uniformly scale the costs
to sum to~$\alpha$: $s'(i,e_j) = \alpha s(i,e_j)/\sum_{k \in V(G)} s(k,e_j)$
(Equilibria will always have $s = s'$.)
Then the cost incurred by player $i$ is
$
c_i(s) = \sum_{e_j \in E_s} s'(i,e_j) + \sum_{j=1}^n d_{G_s}(i,j).
$
The total cost incurred by joint strategy $s$ is
$
c(s) = \alpha \, |E_s| + \sum_{i=1}^n \sum_{j=1}^n d_{G_s}(i,j).
$

In this cooperative model, the notion of Nash equilibrium is less natural
because it allows only one player to change strategy, whereas a cooperative
purchase in general requires many players to change their strategy.
Therefore we use a stronger notion of equilibrium that allows coalition
among players, inspired by the strong Nash equilibrium of
Aumann \cite{Aumann59}, and modeled after the pairwise stability property
introduced for the bilateral game \cite{CP05}.
Namely, a joint strategy $s$ is \emph{collaboratively equilibrium} if,
for any edge $e$ of the host graph~$G$,
for any coalition $C \subseteq V$,
for any joint strategy $s'$ differing from $s$
in only $s'(i,e)$ for $i \in C$,
some player $i \in C$ has $c_i(s') > c_i(s)$.
Note that any such joint strategy must have every sum
$\sum_{i \in V(G)} s(i,e_j)$
equal to either $0$ or~$\alpha$, so we can measure the cost $c_i(s)$
in terms of $s(i,e_j)$ instead of $s'(i,e_j)$.
The \emph{price of anarchy} is the maximum cost of a collaborative
equilibrium divided by the minimum cost of any joint strategy
(the \emph{social optimum}).

\label{dynamics}
We can define a simple dynamics for the cooperative network creation game
in which we repeatedly pick a pair of vertices, have all players determine
their valuation of an edge between those vertices (change in $c_i(s)$
from addition or removal), and players thereby bid on the edge and change
their strategies.  These dynamics always converge to a collaborative
equilibrium because each change decreases the total cost $c(s)$,
which is a discrete quantity in the lattice $\mathbb Z + \alpha \mathbb Z$.
Indeed, the system therefore converges after a number of steps polynomial
in $n$ and the smallest integer multiple of $\alpha$ (if one exists).
More generally, we can show an exponential upper bound in terms of just $n$
by observing that the graph uniquely determines $c(s)$, so we can never
repeat a graph by decreasing $c(s)$.

\section{Preliminaries}
\label{preliminaries}

In this section, we define some helpful notation and prove some
basic results. Call a graph $G_s$ corresponding to an
equilibrium joint strategy $s$ an \emph{equilibrium graph}.
In such a graph, let
$d_{G_s}(u,v)$ be the length of the shortest path from $u$ to~$v$ and
$\Dist_{G_s}(u)$ be $\sum_{v \in V(G_s)}d_{G_s}(u,v)$. Let $N_k(u)$ denote
the set of vertices with distance at most $k$ from vertex~$u$, and
let $N_k = \min_{v \in G} |N_k(v)|$.
In both the unilateral and cooperative network creation games,
the total cost of a strategy consists of two parts.
We refer to the cost of buying edges as the \emph{\edgecost}
and the cost $\sum_{v \in V(G_s)}d_{G_s}(u,v)$ as the \emph{\pathcost}.

First we prove the existence of collaborative equilibria for
complete host graphs.
Similar results are known in the unilateral case \cite{FLMPS03,AFM07}.

\begin{lemma}
  In the cooperative network creation game,
  any complete graph is a collaborative equilibrium for $\alpha \leq 2$,
  and any star graph is a collaborative equilibrium for $\alpha \geq 2$.
\end{lemma}

Next we show that, in the unilateral version,
a bound on the usage cost suffices to bound the total cost
of an equilibrium graph~$G_s$, similar to \cite[Lemma~1]{PODC07}.

\begin{lemma} \label{tree}
  The total cost of any equilibrium graph in the unilateral game is at most
  $\alpha \, n + 2 \sum_{u,v \in V(G_s)}d_{G_s}(u,v)$.
\end{lemma}

Next we prove a more specific bound for the cooperative version,
using the following bound on the number of edges in a graph of large girth:

\begin{lemma} {\rm \cite{DB91}} \label{sparse}
  The number of edges in an $n$-vertex graph of odd girth $g$
  is $O(n^{1+2/(g-1)})$.
\end{lemma}

\begin{lemma} \label{general_bound}
For any integer $g$,
the total cost of any equilibrium graph $G_s$
is at most $\alpha \, O(n^{1+2/g}) + g \sum_{u,v \in
V(G_s)}d_{G_s}(u,v)$.
\end{lemma}

\section{Cooperative Version in Complete Graphs}
\label{Cooperative Version in Complete Graphs}

In this section, we study the price of anarchy when any number of players
can cooperate to create any link, and the host graph is the complete graph.

We start with two lemmata that hold
for both the unilateral and cooperative versions of the problem.
The first lemma bounds a kind of ``doubling radius'' of large neighborhoods
around any vertex, which the second lemma uses to bound the \pathcost.

\begin{lemma} \label{2k+2alpha/n}
  {\rm \cite[Lemma~4]{PODC07}}
  For any vertex $u$ in an equilibrium graph $G_s$,
  if $|N_k(u)| > n/2$, then $|N_{2k+2 \alpha / n}(u)| \geq n$.
\end{lemma}

\begin{lemma} \label{N_k>n/2}
If we have $N_k(u) > n/2$ for some vertex $u$ in an
equilibrium graph $G_s$, the \pathcost\ is at most $O(n^2k+\alpha n)$.
\end{lemma}

Next we show how to improve the bound on ``doubling radius'' for
large neighborhoods in the cooperative game:

\begin{lemma} \label{2k+4sqrt(alpha/n)}
  For any vertex $u$ in an equilibrium graph $G_s$,
  if $|N_k(u)| > n/2$, then $|N_{2k+4 \sqrt{\alpha / n}}(u)| \geq n$.
\end{lemma}

Next we consider what happens with arbitrary neighborhoods,
using techniques similar to \cite[Lemma~5]{PODC07}.

\begin{lemma} \label{4k+3}
  If $|N_k(u)| \geq Y$ for every vertex $u$ in an equilibrium graph~$G_s$,
  then either $|N_{4k+2}(u)| > n/2$ for some vertex $u$
  or $|N_{5k+3}(u)| \geq Y^2 n/\alpha$ for every vertex~$u$.
\end{lemma}

\begin{proof}
  If there is a vertex $u$ with $|N_{4k+2}(u)| > n/2$, then the
  claim is obvious.  Otherwise, for every vertex~$u$, $|N_{4k+2}(u)|
  \leq n/2$.  Let $u$ be an arbitrary vertex.  Let $S$ be the set of
  vertices whose distance from $u$ is $4k+3$.
  We select a subset of $S$, called \emph{center points},
  by the following greedy algorithm.
  %First we unmark all vertices in~$S$.
  We repeatedly select an unmarked vertex $z \in S$ as a center point,
  mark all unmarked vertices in $S$ whose distance from $z$ is at
  most $2 k$, and assign these vertices to~$z$.

\begin{wrapfigure}{r}{3in}
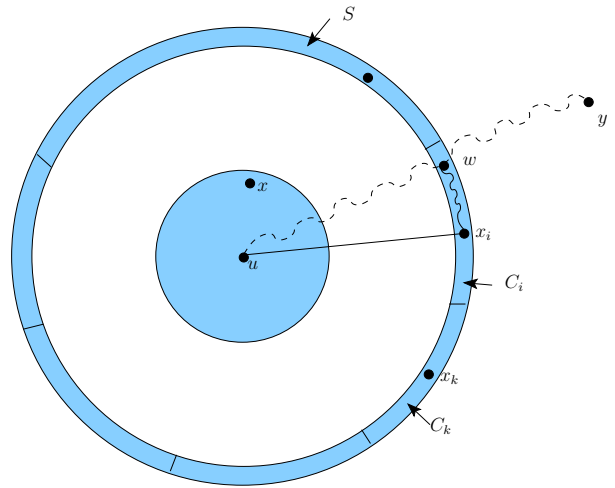

\centering
\scalebox{0.6}{\ifpdf\input{lemma8_1.pdf_t}\else\input{lemma8_1.pstex_t}\fi}
\caption{Center points.}
\label{fig:02}
\vspace{-2ex}
\end{wrapfigure}

  Suppose that we select $l$ vertices $x_1, x_2, \dots, x_l$ as center points. 
  We prove that $l \geq |N_k(u)| n/\alpha$.
  Let $C_i$ be the vertices in $S$ assigned to~$x_i$; see
  %The center points and sets $C_1, C_2, ..,C_l$ are shown in
  Figure~\ref{fig:02}. 
  By construction, $S = \bigcup_{i=1}^l C_i$.
  We also assign each vertex $v$ at distance at least $4 k + 4$ from $u$
  to one of these center points, as follows.
  Pick any one shortest path from $v$ to~$u$
  that contains some vertex $w \in S$,
  and assign $v$ to the same center point as~$w$. This vertex $w$ is
  unique in this path because this path is a shortest path from $v$ to~$u$.
  Let $T_i$ be the set of vertices assigned to $x_i$
  and whose distance from $u$ is more than $4k+2$.
  By construction, $\bigcup_{i=1}^l T_i$ is the set of vertices
  at distance more than $4k+2$ from~$u$.
  The shortest path from $v \in T_i$ to $u$ uses some vertex $w \in C_i$.
  For any vertex $x$ whose distance is at most $k$ from $u$ and
  for any $y \in T_i$, adding the edge $\{u,x_i\}$ decreases the distance
  between $x$ and $y$ at least~$2$, because
  the shortest path from $y \in T_i$ to $u$ uses some vertex $w \in C_i$,
  as shown in Figure~\ref{fig:02}.
  By adding edge $\{u,x_i\}$, the distance between $u$ and $w$
  would become at most $2k+1$ and the distance between $x$ and $w$
  would become at most $3k+1$, where $x$ is any vertex whose distance
  from $u$ is at most $k$. Because the current distance between $x$
  and $w$ is at least $4k+3-k=3k+3$, adding the edge $\{u,x_i\}$ decreases
  this distance by at least~$2$. Consequently the distance between $x$
  and any $y \in T_i$ decreases by at least~$2$. Note that the distance
  between $x$ and $y$ is at least $d_{G_s}(u,y)-k$, and after adding
  edge $(u,x_i)$, this distance becomes at most
  $3k+1 + d_{G_s}(w,y) = 3k+1 + d_{G_s}(u,y) - d_{G_s}(u,w)
                       = 3k+1 + d_{G_s}(u,y) - (4k+3)
                       = d_{G_s}(u,y) - k - 2$.

  Thus any vertex $y \in T_i$ has
  incentive to pay at least $2 \, |N_k(u)|$ for edge $\{u,x_i\}$.
  Because the edge $\{u,x_i\}$ is not in equilibrium, we conclude that
  $\alpha \geq 2 |T_i| |N_k(u)|$.  On the other hand, $|N_{4k+2}(u)| \leq n/2$,
  so $\sum_{i=1}^l |T_i| \geq n/2$.  Therefore, $l \, \alpha \geq
  2 |N_k(u)| \sum_{i=1}^l |T_i| \geq n |N_k(u)|$ and hence $l \geq n |N_k(u)|/\alpha$.

  According to the greedy algorithm, the distance between any pair of
  center points is more than~$2 k$;
  hence, $N_k(x_i) \cap N_k(x_j) = \emptyset$ for $i \neq j$.
  By the hypothesis of the lemma, $|N_k(x_i)| \geq Y$ for every vertex~$x_i$;
  hence $|\bigcup_{i=1}^l N_k(x_i)| = \sum_{i=1}^l |N_k(x_i)| \geq l \, Y$.
  For every $i \leq l$, we have $d_{G_s}(u,x_i)=4k+3$, so vertex $u$
  has a path of length at most $5k+3$ to every vertex whose distance to $x_i$
  is at most~$k$.
  Therefore, $|N_{5k+3}(u)| \geq |\bigcup_{i=1}^l N_k(x_i)| \geq
  l \, Y \geq Y n |N_k(u)|/\alpha \geq Y^2 n/\alpha$.
\end{proof}

Now we are ready to prove bounds on the price of anarchy.
We start with the case when $\alpha$ is a bit smaller than~$n$:

\begin{theorem} \label{1/epsilon^3}
  For $1 \leq \alpha < n^{1-\epsilon}$,
  the price of anarchy is at most $O(1/\epsilon^{1+\lg 5})$.
\end{theorem}

Next we prove a polylogarithmic bound on the price of anarchy
when $\alpha$ is close to~$n$.

\begin{theorem} \label{log(n)^3}
  For $\alpha = O(n)$, the price of anarchy is $O(\lg^{1+\lg 5} n)$
  and the diameter of any equilibrium graph is $O(\lg^{\lg 5} n)$.
\end{theorem}

\begin{proof}
  Consider an equilibrium graph $G_s$.
  The proof is similar to the proof of Theorem~\ref{1/epsilon^3}.
  Define $a_1=\max\{2 , 2 \alpha /n\}+1$ and $a_i=5 a_{i-1}+3$,
  or equivalently $a_i=\frac{4a_1+3}{20} \cdot 5^{i}- {3 \over 4} < a_1 5^i$,
  for all $i>1$.
  %Let $N_k=\min_{v \in V(G_s)}N_k(v)$.
  By Lemma~\ref{4k+3}, for each $i \geq 1$, either $N_{4 a_i+2}(v) > n/2$ for
  some vertex $v$ or $N_{a_{i+1}} \geq (n/\alpha) \, N_{a_i}^2$.
  Let $j$ be the least number for which $|N_{4 a_j+2}(v)| > n/2$ for some
  vertex~$v$.
  By this definition, for each $i < j$,
  $N_{a_{i+1}} \geq (n/\alpha) \, N_{a_i}^2$.
  Because $N_{a_1} > 2 \max\{1,\alpha/n\}$, we obtain that $N_{a_i} > 2^{2^{i-1}}\max\{1,\alpha/n\}$
  for every $i \leq j$.
  On the other hand, $2^{2^{j-1}} \leq 2^{2^{j-1}}\max\{1,\alpha/n\} < N_{a_j} \leq n$,
  so $j < \lg\lg n+1$ and $a_j < a_1 \, 5^{\lg\lg n+1} <
  (2+2\alpha/n+1+1)5\lg^{\lg 5} n = 10(2+\alpha/n)\lg^{\lg 5} n$. Therefore
  $N_{4 \cdot[10(2+\alpha/n)\lg^{\lg 5} n]+2}(v) > n/2$
  for some vertex $v$ and using Lemma
  \ref{2k+2alpha/n}, we conclude that the distance of $v$ to all other vertices
  is at most $2[40(2+\alpha/n)\lg^{\lg 5} n + 2] + 2\alpha/n$. Thus the diameter of $G_s$
  is at most $O((1+\alpha/n)\lg^{\lg 5} n)$.
  Setting $g= \lg n$ in Lemma~\ref{general_bound}, the cost of $G_s$
  is at most $\alpha \, O(n) + (\lg n) O(n^2(1+\alpha/n)\lg^{\lg 5} n)=O((\alpha n +
  n^2)\lg^{1+\lg 5} n)$. Therefore the price of anarchy is at most $O(\lg^{1+\lg 5} n)$.
\end{proof}

When $\alpha$ is a bit larger than $n$, we can obtain a constant bound
on the price of anarchy.  First we need a somewhat stronger result on
the behavior of neighborhoods:

\begin{lemma} \label{5k+1}
  If $|N_k(u)| \geq Y$ for every vertex $u$ in an equilibrium graph~$G_s$,
  then either $|N_{5k}(u)| > n/2$ for some vertex $u$
  or $|N_{6k+1}(u)| \geq Y^2 k n/2\alpha$ for every vertex~$u$.
\end{lemma}

\begin{theorem} \label{sqrt(alpha/n)}
  For any $\alpha > n$,
  the price of anarchy is $O(\sqrt{n/\alpha}\lg^{1+\lg 6} n)$
  and the diameter of any equilibrium graph is
  $O(\lg^{\lg 6} n \cdot \sqrt{\alpha/n})$.
\end{theorem}

By Theorem~\ref{sqrt(alpha/n)}, we conclude the following:

\begin{corollary}
For $\alpha=\Omega(n \lg^{2+2\lg 6} n) \approx \Omega(n \lg^{7.16} n)$,
the price of anarchy is $O(1)$.
\end{corollary}

\section{Cooperative Version in General Graphs}
\label{Cooperative Version in General Graphs}
\label{cooperativeupper}

In this section, we study the price of anarchy when only some links
can be created, e.g., because of physical limitations.
In this case, the social optimum is no longer simply a clique or a star.

We start by bounding the growth of distances from the host graph $G$
to an arbitrary equilibrium graph~$G_s$:

\begin{lemma} \label{k^2/3}
  For any two vertices $u$ and $v$ in any equilibrium graph $G_s$,
  $d_{G_s}(u,v) = O(d_G(u,v) + \alpha^{1/3} d_G(u,v)^{2/3})$.
\end{lemma}

\begin{proof}
Let $u=v_0, v_1, \dots, v_k=v$ be a shortest path in $G$ between $u$
and~$v$, so $k=d_G(u,v)$. Suppose that the distance between $v_0$ and
$v_i$ in $G_s$ is $d_i$, for $0 \leq i \leq k$. We first prove that $d_{i+1}
\leq d_i + 1 + \sqrt{9 \alpha/d_i}$ for $0 \leq i < k$. If edge
$\{v_i,v_{i+1}\}$ already exists in $G_s$, the inequality clearly holds.
Otherwise, adding this edge decreases the distance between $x$ and
$y$ by at least $\frac{d_{i+1}-d_i}{3}$, where $x$ is a vertex whose
distance is at most $\frac{d_{i+1}-d_i}{3}-1$ from $v_{i+1}$ and $y$
is a vertex in a shortest path from $v_i$ to~$v_0$. Therefore any
vertex $x$ whose distance is at most $\frac{d_{i+1}-d_i}{3}-1$ from
$v_{i+1}$ can pay $\frac{d_{i+1}-d_i}{3}d_i$ for this edge.
Because this edge does not exist in $G_s$ and
because there are at least $\frac{d_{i+1}-d_i}{3}$ vertices
of distance at most $\frac{d_{i+1}-d_i}{3}-1$ from~$v_{i+1}$,
we conclude that ${\left(\frac{d_{i+1}-d_i}{3}\right)}^2d_i \leq
\alpha$. Thus we have  $d_{i+1} \leq d_i + 1
+ \sqrt{9 \alpha/d_i}$ for $0 \leq i < k$.
Next we prove that $d_{i+1} \leq d_i + 1 + 5 \alpha^{1/3}$. If edge
$\{v_i,v_{i+1}\}$ already exists in $G_s$, the inequality clearly holds.
Otherwise, adding this edge decreases the distance between $z$ and
$w$ by at least $\frac{d_{i+1}-d_i}{5}$, where $z$ and $w$ are two
vertices whose distances from $v_{i+1}$ and $v_i$, respectively,
are less than $\frac{d_{i+1}-d_i}{5}$. There are at least at least
$\left(\frac{d_{i+1}-d_i}{5}\right)^2$ pair of vertices like
$(z,w)$. Because the edge $\{v_i,v_{i+1}\}$ does not exist
in $G_s$, we conclude that $\left(\frac{d_{i+1}-d_i}{5}\right)^3 \leq
\alpha$. Therefore $d_{i+1} \leq d_i + 1 +5 \alpha^{1/3}$. Combining
these two inequalities, we obtain $d_{i+1} \leq d_i + 1 + \min\{
\sqrt{9 \alpha/d_i} , 5 \alpha^{1/3} \}$.

Inductively we prove that $d_j \leq 3j + 7 \alpha^{1/3} + 5
\alpha^{1/3}j^{2/3}$. For $j \leq 2$, the inequality is clear. Now
suppose by induction that $d_j \leq 3j +7 \alpha^{1/3} + 5
\alpha^{1/3}j^{2/3}$. If $d_j \leq 2
\alpha^{1/3}$, we reach the desired inequality
using the inequality $d_{j+1} \leq d_j + 1 + 5 \alpha^{1/3}$.
Otherwise, we know that
$d_{j+1} \leq d_j + 1 + \sqrt{9 \alpha/d_j} = f(d_j)$
and to find the maximum of the
function $f(d_j)$ over the domain $d_j \in [2\alpha^{1/3},j + 7 \alpha^{1/3} +
5 \alpha^{1/3}j^{2/3}]$, we should check $f$'s critical points, including
the endpoints of the domain interval and where $f$'s derivative is zero.
We reach three values for~$d_j$: $2 \alpha^{1/3}$,  $j + 7
\alpha^{1/3} + 5 \alpha^{1/3}j^{2/3}$, and $\left(\frac{9
\alpha}{4}\right)^{1/3}$. Because the third value is not in the domain,
we just need to check the first two values. The first value is also
checked, so just the second value remains. For the second value,
we have
$d_{j+1} \leq \textstyle d_j + 1 + \sqrt{9 \alpha/d_j}
\leq \textstyle j + 7
\alpha^{1/3} + 5 \alpha^{1/3}j^{2/3} + 1
+ \sqrt{\frac{9 \alpha}{j +
7 \alpha^{1/3} + 5 \alpha^{1/3}j^{2/3}}}
\leq \textstyle j + 1 + 7 \alpha^{1/3} +  5 \alpha^{1/3}j^{2/3} +
\sqrt{\frac{10\alpha}{5\alpha^{1/3}j^{2/3}}}
\leq \textstyle j + 1 + 7 \alpha^{1/3} +  5 \alpha^{1/3}j^{2/3} +
\frac{\alpha^{1/3}\sqrt{2}}{j^{1/3}}$.
Because $(j+1)^{2/3} - j^{2/3} = \frac{(j+1)^2-
j^2}{(j+1)^{4/3}+(j+1)^{2/3}j^{2/3}+j^{4/3}} \allowbreak \geq
\frac{2j}{3(j+1)^{4/3}}$, we have
$j + 1 + 7 \alpha^{1/3} +  5 \alpha^{1/3}j^{2/3} +
\frac{\alpha^{1/3}\sqrt{2}}{j^{1/3}}
\leq \textstyle j + 1 + 7 \alpha^{1/3} + 5 \alpha^{1/3}(j+1)^{2/3}
- 5 \alpha^{1/3}\frac{2j}{3(j+1)^{4/3}} +
\frac{\alpha^{1/3}\sqrt{2}}{j^{1/3}}
\leq \textstyle j + 1 + 7 \alpha^{1/3} + 5 \alpha^{1/3}(j+1)^{2/3}
- \frac{10 \alpha^{1/3} j}{3j^{4/3}} +
\frac{\alpha^{1/3}\sqrt{2}}{j^{1/3}}
\leq \textstyle j + 1 + 7 \alpha^{1/3} + 5
\alpha^{1/3}(j+1)^{2/3}$.

Note that $j+1 > 2$ and $d_k = d_{G_s}(u,v)$. Therefore $d_{G_s}(u,v)$ is at most $O(d_G(u,v) + \alpha^{1/3}d_G(u,v)^{2/3})$ and the desired inequality is proved.
\end{proof}

Using this Lemma~\ref{k^2/3}, we prove two different bounds relating
the sum of all pairwise distances in the two graphs:

\begin{corollary} \label{alpha^1/3}
For any equilibrium graph~$G_s$,
$\sum_{u,v \in V(G)} d_{G_s}(u,v)
= O(\alpha^{1/3}) \cdot \sum_{u,v \in V(G)} d_G(u,v)$.
\end{corollary}

\begin{theorem} \label{n^1/3}
For any equilibrium graph $G_s$,
$\sum_{u,v \in V(G)} d_{G_s}(u,v) \leq
\min\{O(n^{1/3}) (\alpha n + \sum_{u,v \in
V(G)} d_G(u,v)) , n^3\}$.
\end{theorem}

Now we can bound the price of anarchy for the various ranges of~$\alpha$,
combining Corollary~\ref{alpha^1/3}, Theorem~\ref{n^1/3}, and
Lemma~\ref{general_bound}, with different choices of~$g$.

\begin{theorem}
In the cooperative network creation game in general graphs,
the price of anarchy is at most
\begin{enumerate}
\item[\rm (a)]
  $O(\alpha^{1/3})$ for $\alpha < n$~~~[$g=6$ in Lemma~\ref{general_bound} and Corollary~\ref{alpha^1/3}],
\item[\rm (b)]
  $O(n^{1/3})$ for $n \leq \alpha \leq n^{5/3}$~~~[$g=6$ in Lemma~\ref{general_bound} and Theorem~\ref{n^1/3}],
\item[\rm (c)]
  $O(\frac{n^2}{\alpha})$ for $n^{5/3} \leq \alpha < n^{2-\epsilon}$~~~[$g=2/\epsilon$ in Lemma~\ref{general_bound} and Theorem~\ref{n^1/3}], and
\item[\rm (d)]
  $O(\frac{n^2}{\alpha} \lg n)$
  for $n^2 \leq \alpha$~~~[$g=\lg n$ in Lemma~\ref{general_bound} and Theorem~\ref{n^1/3}].
\end{enumerate}
\end{theorem}

\section{Unilateral Version in General Graphs}
\label{Unilateral Version in General Graphs}

Next we consider how a general host graph affects the unilateral
version of the problem.
Some proofs are similar to proofs for the cooperative
version in Section~\ref{cooperativeupper} and hence omitted.

\begin{lemma} \label{k^1/2uni}
  For any two vertices $u$ and $v$ in any equilibrium graph $G_s$,
  $d_{G_s}(u,v) = O(d_G(u,v) + \alpha^{1/2} d_G(u,v)^{1/2})$.
\end{lemma}

Again we relate the sum of all pairwise distances in the two graphs:

\begin{corollary} \label{alpha^1/2uni}
For any equilibrium graph~$G_s$,
$\sum_{u,v \in V(G)} d_{G_s}(u,v)
= O(\alpha^{1/2}) \cdot \sum_{u,v \in V(G)} D_G(u,v)$.
\end{corollary}

\begin{theorem} \label{n^1/2uni}
For any equilibrium graph $G_s$,
$\sum_{u,v \in V(G_s)} d_{G_s}(u,v) \leq
\min\{O(n^{1/2}) (\alpha n + \sum_{u,v \in V(G)} D_G(u,v)) , n^3\}$.
\end{theorem}

To conclude bounds on the price of anarchy, we now use
Lemma \ref{tree} in place of Lemma~\ref{general_bound},
combined with Corollary \ref{alpha^1/2uni} and Theorem \ref{n^1/2uni}.

\begin{theorem}
For $\alpha \geq n$, the price of anarchy is at most
$\min\{O(n^{1/2}) , \frac{n^2}{\alpha}\}$.
\end{theorem}

\begin{theorem}
For $\alpha < n$, the price of anarchy is at most $O(\alpha^{1/2})$.
\end{theorem}

\section{Lower Bounds in General Graphs}
\label{Lower Bounds in General Graphs}

In this section, we prove polynomial lower bounds on the price of anarchy
for general host graphs, first for the cooperative version and second
for the unilateral version.

\begin{theorem} \label{lowerboundcooperative}
The price of anarchy in the cooperative game is
$\Omega(\min\{\sqrt{\frac{\alpha}{n}} , \frac{n^2}{\alpha}\})$.
\end{theorem}

\begin{wrapfigure}{r}{2in}
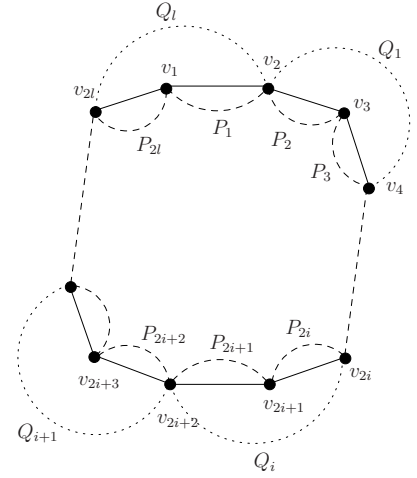

\centering
\vspace{-3ex}
\scalebox{0.65}{\ifpdf\input{lower.pdf_t}\else\input{lower.pstex_t}\fi}
\caption{Lower bound graph.}\label{fig:03}
\vspace{-2ex}
\end{wrapfigure}

\begin{proof}
For $\alpha = O(n)$ or $\alpha = \Omega(n^2)$, the claim is clear.
Otherwise, let $k=\sqrt{\frac{\alpha}{12 n}} \geq 2$.
Thus
$k = O(\sqrt n)$.
We construct graph $G_{k,l}$ as follows; see Figure~\ref{fig:03}.
Start with $2l$ vertices $v_1, v_2, \dots, v_{2l}$ connected in a cycle.
For any $1 \leq i \leq 2l$, insert a path $P_i$ of $k$ edges between $v_i$ and
$v_{i+1}$ (where we define $v_{2l+1}=v_1$).  For any $1 \leq i \leq l$,
insert a path $Q_i$ of $k$ edges between $v_{2i}$ and $v_{2i+2}$
(where we define $v_{2l+2}=v_2$). Therefore there are $n=(3k-1)l$ vertices
and $(3k+2)l$ edges in $G_{k,l}$, so $l = n/(3 k - 1)$.

For simplicity, let $G$ denote $G_{k,l}$ in the rest of the proof.
Let $G_1$ be a spanning connected subgraph
of $G$ that contains exactly one cycle, namely,
$(v_1, v_2, \dots, v_{2l}, v_1)$;
in other words, we remove from $G$ exactly one edge from each path
$P_i$ and~$Q_i$.
Let $G_2$ be a spanning connected subgraph of $G$
that contains exactly one cycle, formed by the concatenation of
$Q_1, Q_2, \dots, Q_l$, and contains none of the edges $\{v_i,v_{i+1}\}$,
for $1 \leq i \leq 2l$;
for example, we remove from $G$ exactly one edge from every
$P_{2i}$ and every edge $\{v_i,v_{i+1}\}$.

Next we prove that $G_2$ is an equilibrium.
For any $1 \leq i \leq l$,
removing any edge of path $Q_i$ increases the distance between its
endpoints and at least $n/6$ vertices by at least
$\frac{l k}{3} \geq n/6$.
Because $\alpha = o(n^2)$, we have $\alpha < \frac{n}{6} \frac{n}{6}$,
so if we assign this edge to be bought solely by one of its endpoints,
then this owner will not delete the edge.
Removing other edges makes $G_2$ disconnected.
For any $1 \leq i \leq l$, adding an edge of path $P_{2i}$ or path
$P_{2i+1}$ or edge $\{v_{2i},v_{2i+1}\}$ or edge $\{v_{2i+1}, v_{2i+2}\}$
to $G_2$ decreases only the distances from some vertices of paths $P_{2i}$
or $P_{2i+1}$ to the other vertices.  There are at most $n(2k-1)$ such pairs.
Adding such an edge can decrease each of these distance by at most $3k-1$.
But we know that $\alpha \geq 12 n k^2 > 2 n (2k-1) (3k-1)$,
so the price of the edge is more than its total benefit among all nodes,
and thus the edge will not be created by any coalition.

The cost of $G_1$ is equal to
$O(\alpha n + n^2 (k+l)) = O(\alpha n + n^2 (k + \frac{n}{k}))$
and the cost of $G_2$ is
$\Omega(\alpha n + n^2(k+l k)) = \Omega(\alpha n + n^3)$.
The cost of the social optimum is at most the cost of $G_1$,
so the price of anarchy is at least
$\Omega(\frac{n^3}{\alpha n + n^3/k +kn^2}) = \Omega(\min\{\frac{n^2}{\alpha} , k, \frac{n}{k}\})$.
Because $k=O(\sqrt{n})$, the price of anarchy is at least
$ \Omega(\min\{\frac{n^2}{\alpha} , k\}) =\Omega(\min\{\frac{n^2}{\alpha} , \sqrt{\frac{\alpha}{n}}\})$.
\end{proof}

\begin{theorem} \label{lower bound unilateral}
The price of anarchy in unilateral games is
$\Omega(\min\{\frac{\alpha}{n} , \frac{n^2}{\alpha}\})$.
\end{theorem}

The proof uses a construction similar to Theorem~\ref{lowerboundcooperative}.

\vspace*{-10ex}

\end{document}